\documentclass[lettersize,journal]{IEEEtran}
\usepackage[style=ieee, backend=biber]{biblatex}
\usepackage{amsmath,amssymb,amsfonts, amsthm}
\usepackage{array}
\usepackage[caption=false,font=normalsize,labelfont=sf,textfont=sf]{subfig}
\usepackage{textcomp}
\usepackage{stfloats}
\usepackage{url}
\usepackage{verbatim}
\usepackage{graphicx}

\usepackage{csquotes}

\usepackage[colorlinks=true, linkcolor=blue, urlcolor=blue, citecolor=blue]{hyperref}
\usepackage{tikz}

\usepackage{caption}
\usepackage{subcaption}

\usepackage{multirow}
\usepackage{booktabs}
\usepackage{diagbox}

\usepackage[english]{babel}
\newtheorem{thm}{Theorem}

\usepackage{algorithm}
\usepackage[noend]{algpseudocode}

\begin{document}

\title{Efficient Homomorphically Encrypted Convolutional Neural Network Without Rotation}

\author{\IEEEauthorblockN{Sajjad Akherati and Xinmiao Zhang}\\
		\IEEEauthorblockA{Department of Electrical and Computer Engineering\\
The Ohio State University, OH 43210, U.S.\\
Emails: \{akherati.1, zhang.8952\}@osu.edu
	}
}



\maketitle

\begin{abstract}
Privacy-preserving neural network (NN) inference can be realized using homomorphic encryption (HE), which enables computations to be performed directly on encrypted data. Popular HE schemes are built over large polynomial rings. To allow simultaneous multiplications in the convolutional (Conv) and fully-connected (FC) layers, multiple input data are encoded into the coefficients of the same polynomial, so are the weights of NNs. However, ciphertext rotations are necessary to compute the sums of products and/or to aggregate outputs from different channels into the same polynomials. Ciphertext rotations have much higher complexity than ciphertext multiplications and contribute to the majority of the latency of HE-evaluated Conv and FC layers. This paper proposes a novel reformulated server-client joint computation procedure and a new filter coefficient packing scheme to eliminate ciphertext rotations without compromising the security of the HE scheme. Our proposed scheme also leads to substantial reductions in the number of coefficient multiplications needed and the communication cost between the server and the client. For various variants of plain-20 classifier over the CIFAR-10/100 datasets, our approach reduces the execution time of the Conv and FC layers by $\mathbf{11.8}\times$, along with decreasing the communication cost between client and server by $\mathbf{40}$\%, compared to the best prior solution.  
\end{abstract}


\section{Introduction}
Convolutional neural networks (CNNs) have enabled the categorization of data with outstanding accuracy and have been applied across a wide range of domains, such as medical diagnosis \parencite{NN_survey, medical_diagnosis1, medical_diagnosis2}, facial recognition\parencite{facial_recognition}, financial data analysis \parencite{finance}, and cyber attack detection \parencite{cybersecurity1, cybersecurity2}. Building NN models demands significant resources. Consequently, NN inference as a service offered by cloud computing \parencite{cloud_computing} is gaining popularity. One of the main drawbacks of this solution is the privacy concerns of users. Privacy-preserving NN inference is enabled by homomorphic encryption (HE) \parencite{FHE}, which allows computations to be directly carried out over encrypted user data in the server. The results, which are also encrypted, are sent back to the user. Then, the user recovers the plaintext results after decryption.  

Popular HE schemes, such as BGV \parencite{BGV}, B/FV \parencite{BV, FV}, and CKKS \parencite{CKKS} schemes, are built over large polynomial rings. Each ciphertext consists of two polynomials, each containing thousands of coefficients with hundreds of bits per coefficient. The complexity of modular polynomial multiplication can be mitigated by the Number Theoretic Transform (NTT) \parencite{HanhoNTT, ParhiNTT}, as well as by integrating the modular reduction into the decomposed Karatsuba algorithm\parencite{PolyMultSiPS, PolyMultJourn}. Modular multiplications on the polynomial coefficients can be simplified by Barrett reduction \parencite{Barrett, ParhiFlexBarrett}, Montgomery algorithm \parencite{Montgomery}, decomposed Karatsuba multiplication \parencite{ZhangMultSiPS, ZhangMult}, and the Chinese Reminder Theorem (CRT) \parencite{ RNSCKKS}. The overall ciphertext multiplication using CRT has been further reformulated to combine computations in \parencite{Combined,Combined2}.  

 Efforts have also been spent toward simplifying the implementation of homomorphically encrypted CNNs. The earlier designs \parencite{cryptonets, CryptoDL, FasterCryptoNets, NGraph-HE, NGraph-HE2} pack data from a large number of inputs into the same polynomial. These designs achieve high throughput but incurred very long latency. To address this, the inputs to the convolutional (Conv) or fully connected (FC) layers are vectorized in the Gazelle design \parencite{Gazelle}, and the corresponding compilers are developed in \parencite{CHET, EVA, EVAimproved}. All these designs require ciphertext rotations to carry out the evaluation of Conv layers and FC layers of CNNs. To mitigate the overhead from rotations, the input is packed multiple times into the same ciphertexts in Lola \parencite{LoLa}, Falcon \parencite{FALCON_CVPR}, and ENSEI \parencite {ENSEI}. These designs exhibit lower latency than Gazelle, particularly when the number of CNN layers is much smaller than the polynomial degree. Cheetah \parencite{Cheetah} extracts the results of convolution from the coefficients of products of polynomials directly. The results of each channel are located in a different ciphertext, and all ciphertexts are sent to the client. Although no ciphertext rotation is needed, it incurs high communication overhead due to the large number of ciphertexts sent from the server to the client. A server-client protocol was recently proposed in Nimbus \cite{nimbus} that enables efficient HE evaluation of Transformers by sending encrypted weights to the client, thereby reducing the communication and computation cost of matrix multiplication. In CNNs, the kernels are much smaller, and the batch convolution involves the dot product of filters with windowed inputs as the window moves across the input image for different output pixels. As a result, the Nimbus approach becomes inefficient in this case. ConvFHE \parencite{ConvFHE} adopts the same idea in Cheetah for computing the convolutions. However, the results of different channels are packed into the same ciphertext utilizing a process that needs ciphertext rotations.


The activation functions of CNNs, such as the rectified linear unit (ReLU), are nonlinear. They can be implemented using either two-party computation (2-PC) protocols, such as Yao's Garbled Circuits (GC) \parencite{GC} and Oblivious Transfer (OT) \parencite{OT}, or approximations over the ciphertexts \parencite{Minimax_approx}. The GC protocol has been simplified in \parencite{FreeXOR, RowReduction, HalfGate, fixed-key-cipher, TinyGarble} and optimizations of OT have been developed in \cite{OT-Extend, VOLE-OT}. Higher-order approximations of ReLU \parencite{Minimax_approx} reduce the precision loss of CNNs but at the same time increase the complexity of evaluation. Although ReLU can also be expressed as Boolean logic functions and implemented directly by using the Torus fully (TF-) HE scheme \parencite{TFHE1, TFHE2, SHE}, $100\times$ larger ciphertext size is required to achieve the same level of security.

Without compromising the security of the HE scheme, this paper introduces a novel reformulated server-client joint computation procedure and a new CNN weight coefficient packing scheme to completely eliminate rotations from the evaluations of Conv and FC layers with low communication cost requirement. By transmitting using one instead of both polynomials of each ciphertext, there is a one-to-one correspondence between the coefficients of the polynomial in the ciphertext and those of the plaintext. As a result, individual coefficients can be extracted without involving ciphertext rotations. Moreover, this one-polynomial method also substantially reduces the polynomial multiplication complexity and the amount of data to communicate between the server and client. Our new weight coefficient packing scheme puts the convolution results of different channels for a Conv layer and different entries of the output vector for a FC layer into adjacent slots of output polynomials. Accordingly, no ciphertext rotation is needed either to collect the results from different ciphertexts. Analysis has been carried out to prove that our reformulated procedure does not compromise the security of the HE scheme or increase the noise level in the ciphertexts. Although our method requires the use of OT for implementing activation functions and proceeding to the next layer, it substantially improves the overall performance of CNN inference in terms of runtime and communication overhead. For various plain-20 classifiers over the CIFAR-10/100 datasets, our design reduces the running time of the Conv and FC layers by ${11.8\times}$ and the communication cost between the client and server by ${40}\%$ compared to the ConvFHE design \cite{ConvFHE}, which is among the best prior designs. Furthermore, when evaluating the ResNet-50 classifier over the ImageNet dataset, our method yields a $4.3\times$ reduction in the latency of linear layers and a $17.1\times$ reduction in communication cost compared to Cheetah \parencite{Cheetah}.

The rest of the paper is organized as follows. Section II provides background knowledge. Section III proposes our joint server-client evaluation procedure and new weight coefficient packing scheme. Section IV presents experimental results and comparisons. Conclusions follow in Section V.

\section{Preliminaries}
This section reviews some essential information on the HE, CNN, and previous packing schemes. Similar to the previous designs \parencite{Gazelle, ConvFHE, ENSEI, FALCON_CVPR, SHE, Cheetah, CHET}, it is assumed that the filters of the CNNs are in plaintext for the evaluation. 

\subsection{CKKS Homomorphic Encryption Scheme}
This paper considers the CKKS scheme \parencite{CKKS}, which is more efficient than other HE schemes, such as BGV and B/FV. It is defined over the polynomial ring $\mathcal{R}_Q:=\mathbb{Z}_Q(X)/(X^N+1)$. The polynomial coefficients are integers mod $Q$. Each polynomial has a degree up to $N-1$, and modular reduction by $X^N+1$ is carried out after polynomial multiplications. The modulus $Q$ should have hundreds of bits and $N$ needs to be in the scale of thousands to achieve a sufficient security level before invoking the costly bootstrapping process to reset the noise level.

Let $\mathit{DG}(\sigma^2)$ represent the Gaussian distribution with variance $\sigma^2$. Each user in the CKKS scheme has a secret key and a public key:
\begin{itemize}
    \item \textbf{Secret key $\mathbf{s(X)}$.} It is a polynomial of degree $N-1$ whose coefficients are randomly selected from $\{0, \pm 1\}$ with the sparsity specified by the target level of security.
    \item \textbf{Public Key $\mathbf{pk = (b(X), a(X))\in \mathcal{R}_Q^2}$.} Here $a(X)$ is a random polynomial from the ring $\mathcal{R}_Q$, and $b(X) = -a(X)s(X)+e(X)\mod Q$, where $e(X)$ is a random polynomial in $\mathcal{R}_Q$ whose coefficients follow the $\mathit{DG}(\sigma^2)$ distribution.
\end{itemize}
The ciphertext, $[m]$, of a plaintext polynomial, $m(X)\in \mathcal{R}$ in the CKKS scheme consists of two polynomials, $[m]=(c^m_0(X),c^m_1(X))$. The CKKS encryption, decryption, ciphertext addition, and multiplication are carried out as follows:
\begin{itemize}
    \item {\bf Encryption}: Let $v(X)$ be a polynomial whose coefficients are sampled from 0, 1, -1 with probability 1/2, 1/4, and 1/4, respectively. Generate random polynomials $e_0(X)$ and $e_1(X)$ from $\mathcal{R}_Q$ following the $\mathit{DG}(\sigma^2)$ distribution. Using a large scalar, $\Delta$, which can be a power of two for simplifying the hardware implementation, $m(X)$ is encrypted into 
    \begin{align}
    [m]\!&=\!(c^m_0(X),c^m_1(X))\nonumber\\
    &=\!v(\!X\!)\!\cdot\! pk\!+\!(\lceil\!\Delta m(\!X\!)\!\rfloor\!\!+\!\!e_0(\!X\!), e_1(\!X\!))\!\mod Q. 
        \label{eq:encryption}
    \end{align}
    \item {\bf Decryption}: For the ciphertext $[m]=(c^m_0(X),c^m_1(X))$, the decrypted message is 
    \begin{equation}
        \label{eq:decryption}
    \lfloor \Delta ^{-1}((c^m_0(X)+c^m_1(X)s(X))\mod Q)\rceil.
    \end{equation}
    \item {\bf Ciphertext Addition}: For the ciphertexts $[m_0]=(c^{m,0}_0(X),c^{m,0}_1(X))$ and $[m_1]=(c^{m,1}_0(X),c^{m,1}_1(X))$, their sum is $[m^+]=(c^{m,0}_0(X)+c^{m,1}_0(X), c^{m,0}_1(X)+c^{m,1}_1(X))\mod Q$. 
    \item {\bf Ciphertext-Plaintext Multiplication}: The product of a ciphertext $[m]=(c^m_0(X),c^m_1(X))$ and plaintext $p(X)$ is $[m^*]$ $=(c^m_0(X)p(X),c^m_1(X)p(X))\mod (X^N+1)$.
\end{itemize}

\subsection{Convolutional Neural Networks}
CNNs process input data through a series of linear and non-linear layers to categorize it into one of several possible classes. An example CNN is shown in Fig. \ref{fig:CNN-example}.\\

The linear layers, shown in Fig. \ref{fig:CNN-example} in blue, have two types: Conv and FC layers. Assume that a Conv layer has $c_i$ input channels and $c_o$ output channels. Let the dimension of each input image be $w_i\times h_i$. Each filter of the Conv layer is of dimension $f_w\times f_h$. For one input image, $I$, and one filter, $F$, the computation carried out in the Conv layer is formulated as:
\begin{align}\label{eq:siso_conv_formula}
    \mathbf{Conv}(I,F)_{k,l} \!:=\! (I*F)_{k,l} \!:=\!\! \sum_{\substack{0 \leq k^\prime < f_w \\ 0 \leq l^\prime < f_h}}\!\!F_{k^\prime,l^\prime}\cdot I_{k+k^\prime,l+l^\prime},
\end{align}
where the subscripts $k$ and $l$ denote the row and column indices, respectively, in the output. For stride-$(s_w, s_h)$ and the valid scheme, the output of each convolution has size $(w_o,h_o)$, where $w_o=\lfloor (w_i-f_w+1)/s_w\rfloor$ and $h_o=\lfloor(h_i-f_h+1)/s_h\rfloor$. Let the input of the $m$-th channel be $I^{(m)}\in\mathbb{Z}^{w_i\times h_i}$. Denote the filter for input channel $m$ and output channel $n$ by $F^{(m,n)}\in\mathbb{Z}^{f_w\times f_h}$. For the $n$-th output channel, the convolution results with the images in all input channels are added up to derive the output as:
\begin{align}\label{eq:batch_conv_formula}
    \mathbf{Conv}(I^{(\cdot)}, F^{(\cdot,n)})&:=
    \sum_{0\leq m<c_i} \mathbf{ Conv}(I^{(m)}, F^{(m,n)})\nonumber\\
    &=\sum_{0\leq m<c_i}I^{(m)}*F^{(m,n)}.
\end{align}

A FC layer that has $n_i$ inputs and $n_o$ outputs is specified by a weight matrix $W$ of dimension 
$n_o\times n_i$ and a $n_o$-entry bias vector $B$. Its output for the input vector $I$ is $WI+B$.

The non-linear layers of the CNN are depicted in green color in Fig. \ref{fig:CNN-example}. They include activation and pooling functions. Activation functions operate on each input element individually, while pooling functions reduce the size in the output. The most frequently used activation function is ReLU, whose output for the input $X$ is $Y=(X+sign(X)X)/2$. The most common pooling function is max pooling, and it can be implemented using the ReLU.

\begin{figure*}[t]
		\centering
		\includegraphics[scale=0.9]{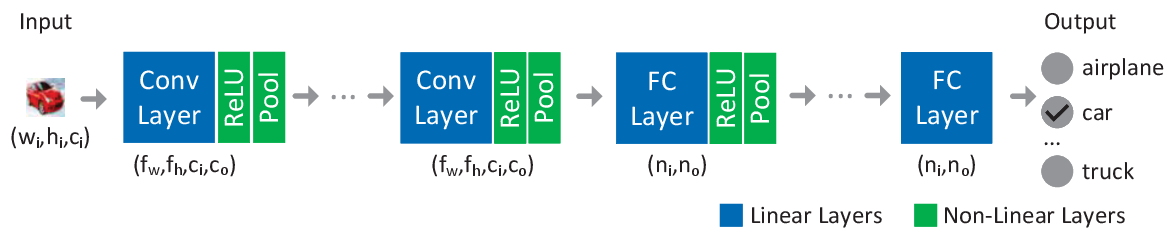}
		\caption{A CNN using the ReLU activation function.}\label{fig:CNN-example}
\end{figure*}

\subsection{Encoding and Packing}
Each ciphertext polynomial in the CKKS scheme has thousands of coefficients. It is unrealistic to encrypt each data into a separate ciphertext. Instead, multiple data can be encoded and packed into the same ciphertext, so that a single ciphertext operation corresponds to computations on multiple plaintext data. However, existing schemes need to add up different slots in the same ciphertext in order to carry out \eqref{eq:siso_conv_formula} and \eqref{eq:batch_conv_formula} or extracting entries from multiple ciphertexts into a single ciphertext. These processes require ciphertext rotations or similar operations, which are much more complicated than ciphertext multiplications and contribute to the majority of HE CNN evaluation latency.
 
\begin{figure}[t]
		\centering
		\includegraphics[scale=0.4]{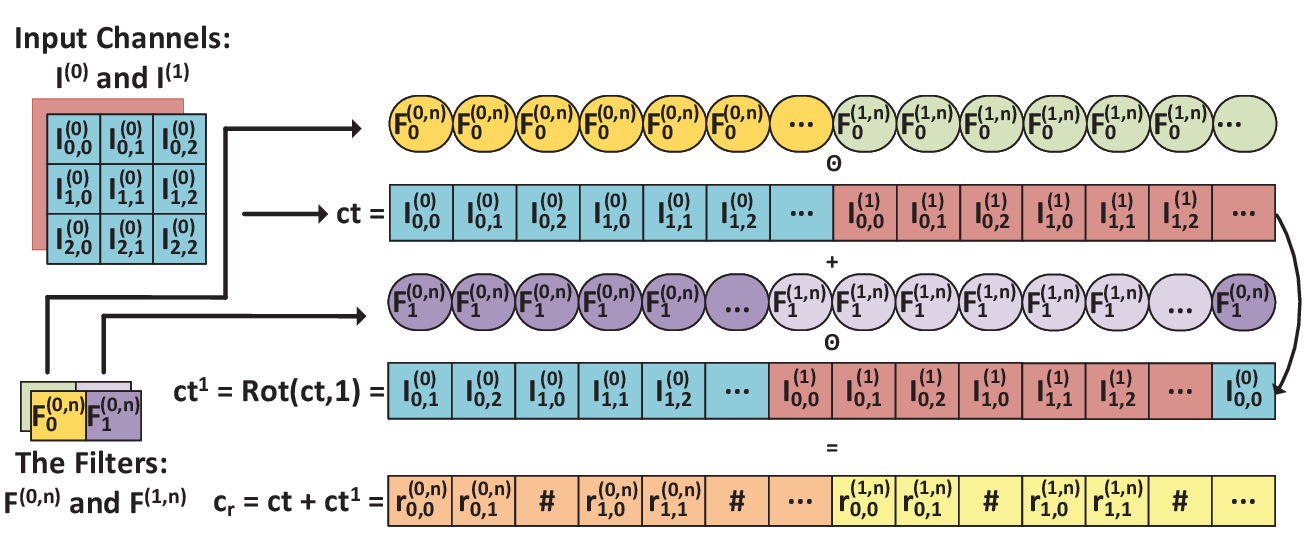}
		\caption{Conv layer using vectorized encoding with $c_i=2$, $w_i=h_i=3$, $f_w=2$, $f_h=1$, stride one, and valid padding for the output channel $n$ .} \label{fig:vec_encoding}
\end{figure}

The design of Gazelle \parencite{Gazelle} introduced the vectorized encoding method, which is also utilized in \parencite{CHET, EVA, EVAimproved}. Fig. \ref{fig:vec_encoding} illustrates an example Conv layer with $c_i=2$, $w_i=h_i=3$, $f_w=2$, $f_h=1$, stride one, and valid padding for the output channel $n$. The plaintext coefficients of the two filters are expanded according to the convolution pattern with stride one and mapped onto different polynomials. The data of two input channels are mapped to the same ciphertext. They are rotated and multiplied with the filter polynomials to compute \eqref{eq:siso_conv_formula}. As shown in the last row of this figure, the results of two different input channels are located in the same ciphertext. It needs to be rotated by $w_ih_i$ slots and added to itself to carry out \eqref{eq:batch_conv_formula}.

A spectral encoding method was employed in \parencite{Falcon,ENSEI, FALCON_CVPR}, which eliminates the need for ciphertext rotations when computing \eqref{eq:siso_conv_formula} by extracting convolution results directly from polynomial multiplications. However, it needs discrete Fourier transform (DFT) over ciphertexts, which leads to accuracy loss and high complexity. The most efficient existing HE convolution was proposed by Cheetah \cite{Cheetah} and a similar scheme was utilized in ConvFHE \parencite{ConvFHE}. Without loss of generality, stride 1 is considered in the following. ConvFHE defines $i^{(m)}(X):=\sum_{0\leq k<w_i,0\leq l<h_i} I^{(m)}_{k,l} X^{(k-f_w)h_i+l}\mod (X^N+1)$ and $f^{(m,n)}(X):=\sum_{0\leq k<f_w, 0\leq l<f_h}\Delta F^{(m,n)}_{k,l} X^{h_if_w-(kh_i+l)}\mod (X^N+1)$. It is assumed that the filter weights have been normalized to the range of $[-1,1]$, and hence they are scaled by $\Delta$ in this packing scheme. Then the $(kh_i+l)$-th coefficient of $i^{(m)}(X)\cdot f^{(m,n)}(X)$ equals $(I^{(m)},F^{(m,n)})_{k,l}$ and ciphertext rotation is not needed to calculate \eqref{eq:siso_conv_formula}. Let $i^{(m), <s>}(X)=i^{(m)}(X^s)$ and $f^{(m,n),<s>}(X) = f^{(m,n)}(X^s)$. Naturally, the $s(kh_i+l)$-th coefficient of $i^{(m),<s>}(X)\cdot f^{(m,n),<s>}(X)$ is $\mathbf{Conv}(I^{(m)},F^{(m,n)})_{k,l}$. To compute the batch convolution for the $n$-th output channel in \eqref{eq:batch_conv_formula}, define 
\begin{align}
    &i(X) = \sum_{0\leq m<c_i} i^{(m),<s>}(X)X^m,\label{eq:i_batch}\\
    &f^{(n)}(X) = \sum_{0\leq m<c_i}f^{(m,n),<s>}(X)X^{-m}\label{eq:k_batch}.
\end{align}
Assuming $N=\text{max}(w_ih_is,f_wf_hs)$, it was shown in \parencite{ConvFHE} that, when $s\geq c_i$, the $s(kh_i+l)$-th coefficient of $r^{(n)}(X)=i(X)\cdot f^{(n)}(X)$ is $\mathbf{ Conv}(I^{(\cdot)}, F^{(\cdot, n)})_{k,l}$ in \eqref{eq:batch_conv_formula}. An example of computing one output batch for one Conv layer with $c_i=4$, $w_i=h_i=2$, $f_w=2$, $f_h=1$, and valid padding for the output channel $n$ in the ConvFHE scheme is illustrated in Fig. \ref{fig:batch_conv_example}. The `\#' in this figure denotes invalid slots. The output of the Conv layer needs to be in the same format as the input to continue with the next layer. To achieve this goal, the PackLWE algorithm \parencite{PackLWE}, which involves computations similar to the ciphertext rotations, is employed in the ConvFHE scheme to put the valid coefficients from different output channels into the same ciphertext. In Cheetah, the ciphertext for each output channel is sent to the client to carry out the non-linear layers using VOLE-style OT \cite{VOLE-OT}, through which the valid coefficients from different output channels are packed into the same ciphertext. Although this approach eliminates the need for rotations, it results in a significant number of ciphertexts being transmitted from the server to the client, leading to high communication overhead..

\begin{figure}[t]
		\centering
		\includegraphics[scale=0.38]{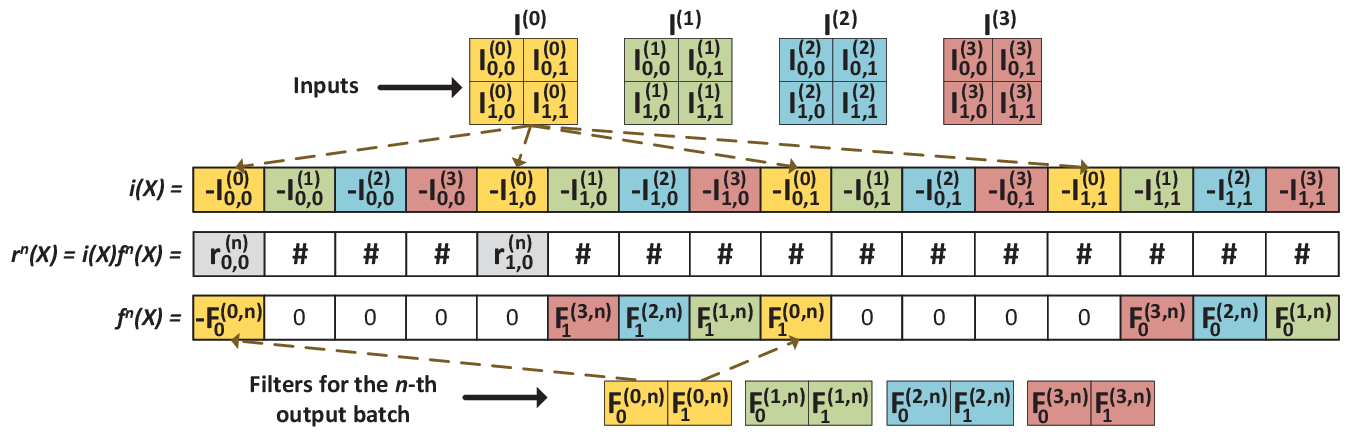}
		\caption{Example of computing one output batch for one convolutional layer with $c_i=4$, $w_i=h_i=2$, $f_w=2$, $f_h=1$, and valid padding for the output channel $n$ in the ConvFHE scheme \parencite{ConvFHE}.} \label{fig:batch_conv_example}
\end{figure}

An FC layer can be described as a weight matrix multiplication. The hybrid scheme in \parencite{Gazelle} maps the entries of the matrix in extended diagonals to the same polynomial. The ciphertext corresponding to the input is rotated and multiplied with the weighted matrix polynomials. The sums of the results go through similar rotate-and-add operations to derive the outputs. Similar to the Conv layer, the Cheetah \cite{Cheetah} and ConvFHE \cite{ConvFHE} designs can eliminate the rotations for weight matrix multiplication. However, the weights are packed into different plaintexts and the multiplication results are also located in different ciphertexts. Putting the results together needs sending all ciphertexts to the client or an extended version of the PackLWE procedure.

\section{HE CNN without Ciphertext Rotation}
This section introduces a reformulated server-client joint computation procedure and a new weight coefficient packing scheme to completely eliminate ciphertext rotations from the HE evaluation of CNNs without incurring large communication costs between the server and client. The proposed scheme can be applied to both Conv and FC layers. This section also provides analysis showing that our reformulated procedure does not compromise the security of the HE scheme or increase the noise level in the ciphertexts.

\subsection{Conv Layer Evaluation without Ciphertext Rotation}
In ConvFHE \parencite{ConvFHE}, the slots whose indices are not multiple of $c_i$ in the product of input ciphertext and filters, $[r^{(n)}(X)]=[i(X)\cdot f^{(n)}(X)]$, as well as those in the plaintext corresponding to the product, $r^{(n)}(X)$, are invalid as shown in the example in Fig. \ref{fig:batch_conv_example}. Invalid slots in $c_1(X)$ can not be directly replaced by zeros. This is because a ciphertext $(c_0(X), c_1(X))$ is decrypted to $\lfloor \Delta ^{-1}((c_0(X)+c_1(X)s(X))\mod Q)\rceil$ according to \eqref{eq:decryption}. If the invalid entries of $c_1(X)$ are replaced by zeros, $c_1(X)s(X)$ will be different and the valid entries in the decryption results will be changed. ConvFHE uses the PackLWE \parencite{PackLWE} algorithm to find another ciphertext whose corresponding plaintext, $r^{'(n)}(X)$, equals $r^{(n)}(X)$ in the valid slots but has negative values in those invalid entries. The PackLWE involves very similar computations as ciphertext rotations, leading to long latency. The plaintext $r^{(n)}(X)+r^{'(n)}(X)$ has zeros in the invalid entries. Hence, rotating $[r^{(n)}(X)]+[r^{'(n)}(X)]$ by $n$ slots and adding up these results for $0\leq n<c_o$, the outputs of different channels are put into adjacent slots in the same ciphertext. It has the same format as the Conv layer input so that the next layer of CNN can continue. 

Our design proposes to extract the valid entries of $[r^{(n)}(X)]$ directly, without changing the decryption result. The intuition is that, from \eqref{eq:decryption}, there is a one-to-one correspondence between the slots in $c_0(X)$ and those in the corresponding plaintext. Hence, the valid slots in $c_0(X)$ can be extracted directly. Although $c_1(X)$ is also required for decryption, it can be handled on the client side. Furthermore, our new design introduces a novel filter packing scheme that eliminates the rotations needed to align the valid slots from different $[r^{(n)}(X)]$ into adjacent slots of the same ciphertext. Our proposed procedure for one Conv layer is outlined in Algorithm \ref{alg:proposed_ext}, assuming that the data from all input channels can be packed into one plaintext polynomial and all Conv layer filters corresponding to one output channel can also be packed into a single plaintext polynomial. In this algorithm, $U(0,\Delta)$ denotes a uniform distribution over the range of $[0,\Delta]$,  and the additions and multiplications on the polynomial coefficients are always followed by modular reductions with the ciphertext modulus. Besides, modular reduction by $X^N+1$ is carried out after every polynomial multiplication. For simplicity, these notations are omitted in Algorithm \ref{alg:proposed_ext}. Our method adopts a modified $c_1(X)$ computation involving the client. The security analysis is presented in Section \ref{sec:security_analysis}. 

\begin{algorithm}[h]
\small	\caption{Proposed HE convolution procedure} \label{alg:proposed_ext}
\begin{algorithmic}[1]
\State \textbf{Client Input}: $s(X), pk, i(X)$
\State \textbf{Server Input}: $pk,\hat f^{(n)}(X)$
        \State \textbf{Server initialization:} 
        \State \indent\hspace{0em} \parbox[t]{\dimexpr\linewidth-\algorithmicindent}{\raggedright 
 Sample random polynomials $\rho^{(n)}(X)$ $(0 \leq n < c_o)$ with coefficients following the $U(0,\Delta)$ distribution;}
        
        \State \indent\hspace{0em} \parbox[t]{\dimexpr\linewidth-\algorithmicindent}{\raggedright 
 Sample random noise polynomials $e^{(n)}_{1}(X)$ and $e_2^{(n)}(X)$ $(0\leq n< c_o)$ with coefficients following the  $DG(\sigma^2)$ distributions;}
        
        \State \indent\hspace{0em} \parbox[t]{\dimexpr\linewidth-\algorithmicindent}{\raggedright

    Compute $p^{(n)}_1(X) =  \rho^{(n)}(X)b(X) + e^{(n)}_{1}(X)$;}
    \State \indent\hspace{0em} \parbox[t]{\dimexpr\linewidth-\algorithmicindent}{Compute $p^{(n)}_2(X) = \hat f^{(n)}(X)+\rho^{(n)}(X)+e^{(n)}_2(X)$;}

        \State \hspace{2em} Send $p^{(n)}_1(X)$ and $p^{(n)}_2(X)$ to the client.\\
        \State \textbf{Encryption (client side)}
        \State \hspace{2em} \parbox[t]{\dimexpr\linewidth-\algorithmicindent}{\raggedright Compute the $c_0(X)$ of $-i(X)$ encryption according to \eqref{eq:encryption};}
        \State \hspace{2em} Store the random $v(X)$ in the client;
        \State \hspace{2em} Send $c_0(X)$ to the server.
        \State {\bf Evaluation (server side)}
        \State \hspace{2em} \textbf{for} $0\leq n< c_o$ \textbf{do}
        \State \hspace{2em} \hspace{2em} $c_0^{(n)}(X)\gets c_0(X)\cdot  \rho^{(n)}(X)-sh(X)\cdot\hat f^{(n)}(X)$
        \State \hspace{2em} \hspace{2em} 
        $c_0^{(n)}(X) \gets \Delta^{-1}c_0^{(n)}(X)$
        \State \hspace{2em} \textbf{for} $j=0$; $j<N$; $j=j+c_i$ \textbf{do}
        \State \hspace{2em} \hspace{2em} \textbf{for} $0\leq n < c_o$ \textbf{do}
        \State \hspace{2em} \hspace{2em} \hspace{2em}  $c^r_{0,j+nc_i/c_o}\gets c^{(n)}_{0, j+nc_i/c_o}$
        \State \hspace{2em} \parbox[t]{\dimexpr\linewidth-\algorithmicindent}{\raggedright Sample random polynomials $sh'(X)$ with coefficients following $U(0,\Delta)$ distribution;}
          \State \hspace{2em} \parbox[t]{\dimexpr\linewidth-\algorithmicindent}{\raggedright Store $sh'(X)$ to be used as the $sh(X)$ for the next layer;}
        \State \hspace{2em} $c^r_0(X) \gets c^r_0(X) + sh^\prime(X)$
        \State  \hspace{2em} Send $c^r_0(X)$ to the client.      
        \State {\bf Evaluation \& decryption (client side)}
        \State \hspace{2em} \textbf{for} $0\leq n< c_o$ \textbf{do}
        \State \hspace{2em}\hspace{2em} $\ c_1^{(n)}(X)\gets v(X)\cdot p^{(n)}_1(X)+i(X)\cdot p^{(n)}_2(X)$
        \State \hspace{2em} \hspace{2em} 
        $c_1^{(n)}(X)\gets \Delta^{-1}c_1^{(n)}(X)$
        \State \hspace{2em} \textbf{for} $j=0$; $j<N$; $j=j+c_i$ \textbf{do}
        \State \hspace{2em} \hspace{2em} \textbf{for} $0\leq n < c_o$ \textbf{do}
         \State \hspace{2em} \hspace{2em}\hspace{2em} $c^r_{1,j+nc_i/c_o}\gets c^{(n)}_{1,j+nc_i/c_o}$
        \State \hspace{2em} $r(X) \gets c_0^r(X)+c_1^r(X)$ 
\end{algorithmic}
\normalsize
\end{algorithm} 

For the first convolutional layer, $I^{(m)}$ $(0\leq m<c_i)$ are packed into $i(X)$ according to \eqref{eq:i_batch}. In succeeding layers, $i(X)$ is the $r(X)$ output by the previous layer. The secret and public keys are generated by the client. The server receives the public key from the client and has the CNN filters, $F^{(m,n)}$ ($0\leq m<c_i, 0\leq n<c_o$), which are packed into $\hat f^{(n)}(X)$ according to \eqref{eq:k_batch_new}. The server also computes $p^{(n)}_1(X)$ and $p^{(n)}_2(X)$ according to Lines 7-8 of Algorithm \ref{alg:proposed_ext}, which are needed to calculate the $c_1(X)$ part of the convolution results in the client without revealing any information about the filters to the client. 

Starting from the plaintext data, the client carries out encryption but only computes the $c_0(X)$ part of the ciphertext. The random polynomial, $v(X)$, generated in the encryption process is stored, since it is needed to compute the $c_1(X)$ part from $p^{(n)}_1(X)$ as shown in Line 27 of Algorithm \ref{alg:proposed_ext}.

To prevent the client from getting information about $\hat f^{(n)}(X)$ from $r(X)$, a random share $sh'(X)$ is added to $c_0^r(X)$ at each layer, except the last one \parencite{Gazelle}. To offset the contribution of this extra share in the evaluation of subsequent layers, $sh'(X)$ is stored and later used as $sh(X)$ in Line 16 of the Algorithm. Using the secret share $sh(X)$, the $r(X)$ output from Algorithm \ref{alg:proposed_ext} actually equals the evaluation result of the layer added with $sh(X)$. Hence, the input to the second and later layers actually is $i(X)+sh(X)$. Encrypting the negation of this input using \eqref{eq:encryption}, to obtain $c_0(X)$, and then plugging it into the formula in Line 16, taking into account the rescaling in Line 17 and random share addition in Line 23 of Algorithm \ref{alg:proposed_ext}, it can be derived that
\begin{align}
    c_0^{(n)}(\!X\!) \!\!&=\!\! \Delta^{-1}\!(\!-i(\!X\!)\rho^{(n)}(\!X\!)\!-\!sh(\!X\!)\rho^{(n)}\!(\!X\!)\!-\!sh(\!X\!)\!\hat{f}^{(n)}(\!X\!)\nonumber\\
    &\quad +\!v(\!X\!)\rho^{(\!n\!)}(\!X\!)a(\!X\!)s(\!X\!)\!+\!\rho^{(\!n\!)}\!(\!X\!)e_1(\!X\!)\!)\!+\!sh'(\!X\!). \label{eq: c0^(n)(x)}
\end{align}
Equivalently, it can also be derived that, following Lines 27 and 28 of Algorithm \ref{alg:proposed_ext}, at the end of the server evaluation,
\begin{align}
    c_1^{(n)}(X) 
    &= \Delta^{-1}(i(X) \hat{f}^{(n)}(X)-v(X)\rho^{(n)}(X)s(X)a(X)\nonumber\\
    &\quad +\! i(\!X\!)\rho^{(n)}\!(\!X\!)\!+\!sh(\!X\!)\hat f^{(n)}(\!X\!)\!+\!sh(\!X\!)\rho^{(n)}(\!X\!)\nonumber\\
    &\quad+v(X)r^{(n)}(X)e(X)+v(X)e_{1}^{(n)}(X)\nonumber\\
    &\quad + i(X)e_2^{(n)}(X)+sh(X)e_2^{(n)}(X)). \label{eq: c1^(n)(x)}
\end{align}
Adding up the above two equations,
\begin{align}
    c_0^{(n)}(X) + c_1^{(n)}(X) &=\Delta^{-1}(i(X)\hat f^{(n)}(X) + \rho^{(n)}(X)e_1(X)\nonumber \\
    &\quad + v(X)\rho^{(n)}(X)e(X)+v(X)e_{1}^{(n)}(X)\nonumber\\
    &\quad + i(X)e_2^{(n)}(X)+sh(X)e_2^{(n)}(X))\nonumber\\
    &\quad+sh^\prime(X). \label{eq: c0+c1}
\end{align}
In the parenthesis of the above formula, the other terms are much smaller than $i(X)\hat f^{(n)}(X)$, and are considered as the approximation errors. Hence, $\Delta^{-1}(i(X)\hat f^{(n)}(X))+sh'(X)$ is recovered. This expression represents the scaled output of the layer added up with the secret share $sh'(X)$.

In CNNs, $c_o\leq c_i$ because the Conv layers extract features.
Assuming $c_i$ is an integer multiple of $c_o$, our scheme packs the coefficients of all the filters belonging to the same output channel into
\begin{align}
    \hat f^{(\!n\!)}(\!X\!) \!=\!\! \sum_{0\leq m<c_i}\!\!f^{(\!m\!,\!n\!),<\!c_i\!>}(\!X\!)X^{n\tfrac{c_i}{c_o}-m} \!\!=\! X^{n\tfrac{c_i}{c_o}}f^{(n)}(\!X\!). \label{eq:k_batch_new}
\end{align}
As a result, the output of channel $n$ is
\begin{align}
    r^{(n)}(X) \!=\! i(X)\cdot \hat f^{(n)}(X) \!=\! X^{nc_i/c_o}i(X)\cdot f^{(n)}(X).
\end{align}
Since the valid entries of $i(X)\cdot f^{(n)}(X)$ are located at slots $kc_i$ ($0\leq k<N/c_i$) in the ConvFHE design, our new filter coefficient packing scheme arranges the valid coefficients of $i(X)\hat f^{(n)}(X)$ at slots $kc_i+nc_i/c_o$. Accordingly, the encrypted valid coefficients are also located at slots $kc_i+nc_i/c_o$ in the polynomials $c_0^{(n)}(X)$. As a result, no rotation is needed, and these coefficients can be directly collected, as shown in Lines 18-20 of Algorithm \ref{alg:proposed_ext}, to form a ciphertext polynomial $c^r_0(X)$, whose corresponding plaintext polynomial consists of the results of different output channels. In Line 20, $c^r_{0,l}$ and $c^{(n)}_{0,l}$ represent the coefficients of $X^l$ in $c^r_0(X)$ and $c^{(n)}_0(X)$, respectively.  Next, the resulting ciphertext is sent from the server to the client. When $c_o<c_i$, in every group of $c_i$ slots in $c_0^r(x)$, $c_i-c_o$ slots are invalid and and can be omitted during transmission to reduce communication cost.


Using our new packing scheme to generate $\hat f^{(n)}(x)$, the polynomial $c_1^{(n)}(X)$, computed in Line 28 of Algorithm \ref{alg:proposed_ext}, also contains valid coefficients at slots $kc_i+nc_i/c_o$ ($0\leq k<N/c_i$). The valid coefficients are collected to form $c^r_1(X)$. It is then added to $c^r_0(X)$ to obtain $r(X)$, which serves as the plaintext polynomial for the convolution result and it has packed the output in the same format as the input polynomial $i(X)$. Specifically, the ($kc_i+nc_i/c_o$)-th coefficient of $r(X)$  corresponds to the same coefficient of $c_0^{(n)}(X)+c_1^{(n)}(X)$ in \eqref{eq: c0+c1}. 

\begin{figure}
     \centering
     \includegraphics[scale=0.55]{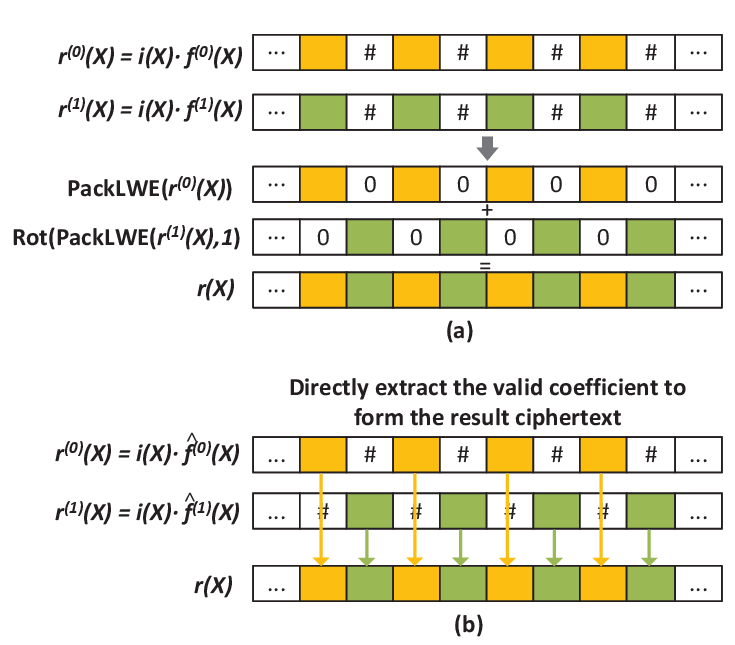}
     \caption{Locations of valid coefficients and formations of convolution result ciphertext in a) ConvFHE scheme \parencite{ConvFHE}; b) our proposed scheme for the case of $c_i=c_o=2$. (`\#' denotes invalid coefficients.)}
     \label{fig:batch_conv_distribution}
\end{figure}

Only the valid slots of $c_0^{(n)}(X)$ are taken to form $c^r_0(X)$ in Line 20 of Algorithm \ref{alg:proposed_ext}. Hence, only those slots need to be computed from the polynomial multiplication of $c_0(X)\rho^{(n)}(X)$ and $sh(X)\hat f^{(n)}(X)$ in Line 16.  Similarly, the polynomial multiplications for generating $c_1^r(X)$ in Line 27 of Algorithm \ref{alg:proposed_ext} can also be simplified.

Fig. \ref{fig:batch_conv_distribution} illustrates the differences in the locations of valid slots and how they are collected into the same ciphertext in our proposed scheme comapred to the ConvFHE design, for the case where $c_i=c_o=2$. Unlike the ConvFHE design, our proposed scheme does not require the expensive PackLWE algorithm to zero out invalid coefficients, nor does it need ciphertext rotations to collect the valid coefficients from multiple ciphertexts.

\subsection{FC Layer Evaluation without Ciphertext Rotation}
The two-dimensional output generated from a Conv layer is flattened into a one-dimensional vector before being sent to an FC layer as input. Multiplying the input ciphertext and weights, the ConvFHE design \parencite{ConvFHE} can still extract the outputs of FC layers from the results of polynomial multiplication. However, it leverages a generalization of the expensive packLWE algorithm \parencite{PackLWE}. Cheetah \cite{Cheetah} sends the ciphertext multiplication results to the client to have the FC layer outputs extracted, incurring high communication costs. Next, a new FC layer weight matrix packing scheme is proposed, so that the results of FC layers are computed without any ciphertext rotation or high communication cost.

Assume that the FC layer has $n_i$ input neurons and $n_o$ output neurons, and  $N=n_in_o$. If $n_in_o<N$, the weight matrix can be padded with zeros to match the size of $N$. If $n_in_o> N$, then the weight matrix can be decomposed into sub-matrices, each of size smaller than or equal to $N$, and our proposed packing scheme remains applicable to these decomposed matrices.


Denote the weight matrix of the FC layer by $W$. Represent the input, output, and bias vectors of the FC layer by $I$, $R$, and $B$, respectively. Our design proposes to pack the input vector and weight matrix into 
\begin{align}
     &i(X) = \sum_{l=0}^{n_i-1}I_{l}X^{ln_o},\label{eq:i(x) formula}\\
     &w(X) = \sum_{k=0}^{n_o-1}W_{k,0}X^k-\sum_{k=0}^{n_o-1}\sum_{l=1}^{n_i-1}W_{k,n_i-l}X^{ln_o+k}. \label{eq:w(x) formula}
\end{align}

\begin{thm}\label{thm:fc for sparse} $R_k$ equals the coefficient of $X^k$ in $(i(X)w(X)+d(X)) \mod (X^N+1)$, where $d(X) = \sum_{k=0}^{n_o-1}B_kX^k$. 
\end{thm}

\begin{figure}[t]
     \centering
     \includegraphics[scale=0.5]{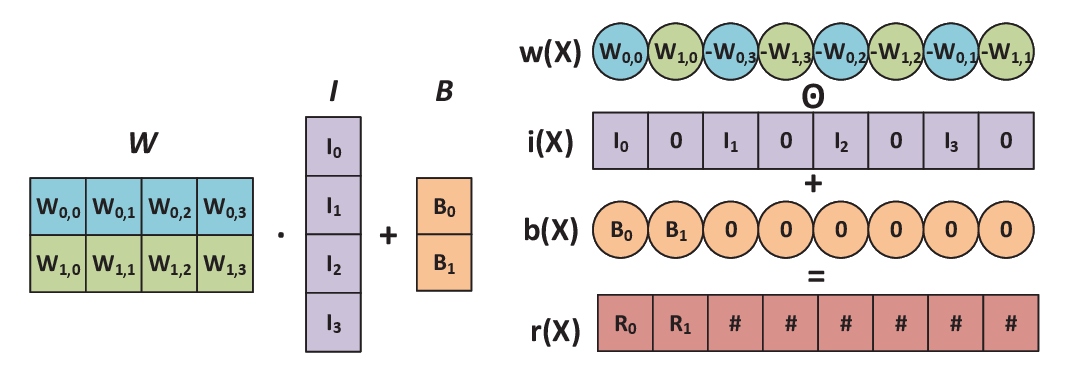}
     \caption{Example of FC layer evaluation using our proposed 
 packing for the case of $n_i=4$, and $n_o=2$. (`\#' denotes invalid slots) }
     \label{fig:fc_eval}
\end{figure}

\begin{proof} Denote the coefficient of $X^k$ in $i(X)$ by $i_k$, and use similar notations for the coefficients in other polynomials. Following polynomial multiplication, the $k$-th coefficient of $i(X)w(X)\mod (X^N+1)$ is
\begin{align}
    \rho_k=\sum_{l=0}^{k}i_lw_{k-l}-\sum_{l=k+1}^{N-1}i_lw_{N-l+k} \label{eq:rho_k formula}.
\end{align}
Only the $\rho_k$ with $0\leq k<n_o$ are needed to compute the FC layer output. From the definition in \eqref{eq:i(x) formula}, $i_0=I_0$ and $i_l=0$ for $1\leq l \leq k < n_o$. Also, $w_k = W_{k,0}$ from \eqref{eq:w(x) formula}. Hence, the first summation on the right hand-side of \eqref{eq:rho_k formula} reduces to $I_0W_{k,0}$. Similarly, from \eqref{eq:i(x) formula}, only $i_l$ with $l=jn_o$ $(j\in Z^+, 1\leq j<n_i)$ are nonzero. Hence, the second summation in \eqref{eq:rho_k formula} becomes
\begin{align*}
    \sum_{l=k+1}^{N-1}i_lw_{N-l+k} = \sum_{j=1}^{n_i-1}i_{jn_o}w_{n_on_i-jn_o+k}.
\end{align*}
From our definition in \eqref{eq:i(x) formula} and \eqref{eq:w(x) formula}, $i_{ln_o} = I_l$ and $w_{n_on_i-ln_o+k}=-W_{k,l}$. Accordingly, \eqref{eq:rho_k formula} reduces to 
\begin{align*}
    \rho_k = I_0W_{k,0}-\sum_{l=1}^{n_i-1}I_l(-W_{k,l}) = \sum_{l=0}^{n_i-1}I_{l}W_{k,l}.
\end{align*}
Adding up with $B_k$, $R_k=\sum_{l=0}^{n_i-1}I_{l}W_{k,l}+B_k$, which equals the $k$-th output of the FC layer. 
\end{proof}

\begin{algorithm}[h]
\small	\caption{Proposed FC evaluation procedure with HE} \label{alg:proposed_fc}
\begin{algorithmic}[1]
        \State \textbf{Client Input}: $s(X), pk, i(X)$
        \State \textbf{Server Input}: $pk,w(X), d(X)$  
        \State \textbf{Server initialization:} 
        \State \indent\hspace{0em} \parbox[t]{\dimexpr\linewidth-\algorithmicindent}{Sample the random polynomials $\rho(X)$ with coefficients following the $\mathbin{U}(0,\Delta)$ distribution;}
        
        \State \indent\hspace{0em} \parbox[t]{\dimexpr\linewidth-\algorithmicindent}{Sample the noise polynomials $e_{w_1}(X)$ and $e_{w_2}(X)$ with coefficients following the $DG(\sigma^2)$ distributions;}
        \State \indent\hspace{0em} \parbox[t]{\dimexpr\linewidth-\algorithmicindent}{Compute $p_1(X) =  \rho(X)b(X) + e_{w_1}(X)$;}
        
        \State \indent\hspace{0em} \parbox[t]{\dimexpr\linewidth-\algorithmicindent}{Compute $p_2(X)=\rho(X)+w(X)+e_{w_2}(x)$;}

        \State \hspace{2em} Send $p_1(X)$ and $p_2(X)$ to the client.\\
        \State \textbf{Encryption (client side)}
        \State \hspace{2em} \parbox[t]{\dimexpr\linewidth-\algorithmicindent}{Compute the $c_0(X)$ of $-i(X)$ encryption according to \eqref{eq:encryption};}
        \State \hspace{2em} Store the random $v(X)$ in the client;
        \State \hspace{2em} Send $c_0(X)$ to the server.
        \State {\bf Evaluation (server side)}
        \State \hspace{2em} $c_0^r(X)\gets \Delta^{-1}(c_0(X)\cdot r(X) - sh(X)\cdot w(X))$
        \State \hspace{2em} \parbox[t]{\dimexpr\linewidth-\algorithmicindent}{\raggedright Sample random polynomials $sh'(X)$ with coefficients following $U(0,\Delta)$ distribution;}
          \State \hspace{2em} \parbox[t]{\dimexpr\linewidth-\algorithmicindent}{\raggedright Store $sh'(X)$ to be used as the $sh(X)$ for the next layer;}
        \State \hspace{2em} $c_0^r(X)\gets c_0^r(X) + d(X) + sh^{\prime}(X)$
        \State  \hspace{2em} Send $c^r_0(X)$ to the client.
        \State {\bf Evaluation \& decryption (client side)}
        \State \hspace{2em} $c_1^{r}(X)\gets \Delta^{-1}(v(X)\cdot p_1(X)+i(X)\cdot p_2(X))$
        \State \hspace{2em} $r(X) \gets c_0^r(X)+c_1^r(X)$
\end{algorithmic}
\normalsize
\end{algorithm} 

\begin{figure*}[t]
		\centering
		\includegraphics[scale=0.37]{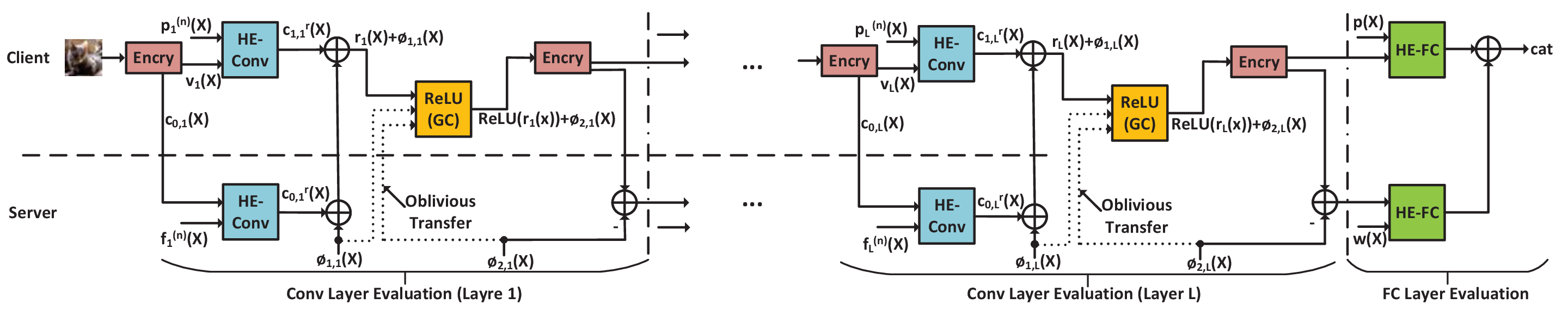}
		\caption{Block diagram of our proposed scheme utilizing VOLE-style OT for evaluating ReLU function.}\label{fig:ReLU_GC}
\end{figure*}

The procedure for evaluating an FC layer using our proposed packing scheme is summarized in Algorithm \ref{alg:proposed_fc}. The weight matrix is packed into the polynomial $w(X)$ as described in \eqref{eq:w(x) formula}. Similar to the evaluation of the proposed Conv layer, the polynomials $p_1(X)$ and $p_2(X)$ need to be computed on the server at initialization and then transmitted to the client to carry out the decryption of the FC layer output. An FC layer takes the output of either a Conv layer or a preceding FC layer as input, where the valid entries are adjacent. However, the valid entries in the FC layer's input must be spaced $n_o$ slots apart. 
Fortunately, decryption is carried out on the client as part of the evaluation of each layer, and the polynomial format can be easily adjusted in the plaintext. Hence, the coefficients of $i(X)$ are packed in the client as per \eqref{eq:i(x) formula}. Since $sh(X)$ is inherently part of $i(X)$, the coefficients of $sh(X)$ must be swapped accordingly on the server side.

Using our proposed packing in \eqref{eq:i(x) formula} and \eqref{eq:w(x) formula}, two polynomial multiplications and three polynomial additions are needed to compute the $c_0^r(X)$ component of the FC layer output, as detailed in Lines 15-18 of Algorithm \ref{alg:proposed_fc}. Next, $c_1^r(X)$ is calculated in the client, as specified in Line 21 of Algorithm \ref{alg:proposed_fc}. Similar to the procedure outlined in Algorithm \ref{alg:proposed_ext} for the Conv layer, $sh'(X)$ is not included in the last FC layer. It can be derived that $r(X) = \Delta^{-1}(i(X)w(X))  + d(X)$, which represents the desired FC layer output. Unlike previous designs, our algorithm does not require any ciphertext rotation. Fig. \ref{fig:fc_eval} illustrates an example of FC layer evaluation using our proposed packing scheme for the case of $n_i=4$ and $n_o=2$. 

For a FC layer with $n_o$ outputs, the first $n_o$ coefficients of $r(X)$ generated by Algorithm \ref{alg:proposed_fc} are the outputs. Therefore, only the first $n_o$ coefficients need to be computed for $c^r_0(X)$ and $c^r_1(X)$. Since only the first $n_o$ coefficients of $c^r_0(X)$ are transmitted from the server to the client, the total communication cost is $N+n_o$ coefficients, including the $N$ coefficients of $c_0(X)$ sent from the client to the server. Although $p_1(X)$ and $p_2(X)$ are sent from the server to the client at initialization, this is done only once for all evaluations and is not included in the communication cost.


\subsection{Activation Function Evaluation with Oblivious Transfer}
A nonlinear activation function typically follows each Conv layer and the most popular activation function is ReLU. Nonlinear functions can not be directly implemented homomorphically. Approximations are adopted in \parencite{ConvFHE, CHET, EVAimproved, TFHE_DNN1} and they cause precision losses on the CNNs. Since the decryption is carried out as part of the proposed evaluation schemes, the millionaire protocol of CryptFlow2 \cite{CrypTFlow2} is utilized to evaluate ReLU with VOLE-style OT \cite{VOLE-OT}, similar to Cheetah \cite{Cheetah}. The overall procedure of evaluating a CNN using our proposed schemes is illustrated in Fig. \ref{fig:ReLU_GC}. 

The evaluation of a Conv layer is carried out according to Algorithm \ref{alg:proposed_ext}. Before the $c_0^r(X)$ of a Conv layer is sent back to the client, it is added with $sh^\prime(X)$, as per Line 23 of Algorithm \ref{alg:proposed_ext}. Then on the client side, $r(X)+sh^\prime(X)$ and $sh^\prime(X)$ are utilized to evaluate ReLU with millionaire protocol via OT \cite{Cheetah, CrypTFlow2, VOLE-OT}. The protocol outputs $\text{ReLU}(r(X))+sh^{\prime\prime}(X)$, which is the input to the next Conv or FC layer. 


\subsection{Security and Noise Analyses}\label{sec:security_analysis}
This section analyzes the potential impacts on the security and noise brought by our proposed schemes. Specifically, our protocol involves sharing polynomials computed from the CNN filters and weights with the client. It is demonstrated that the shared information does not leak any of the data about the CNN to the client. Moreover, our proposed scheme does not require any rotations, thereby eliminating the noise introduced by rotation operations. As a result, our design achieves a lower noise level compared to the original HE scheme.

In our proposed algorithms, the server receives a polynomial $c_0(X)$, which is the first element of the ciphertext obtained from the encryption of the client's input. Under the semantic security of the CKKS HE scheme, the user's data is not exposed to the server. Furthermore, a secret share, $sh(X)$, known only to the server, is added to the evaluation result. Because the evaluation output $i(x)$ has coefficients in the range $[0, \Delta]$, the coefficients of $sh(X)$ must also lie within this range to maintain security \cite{Gazelle}. Accordingly, the client does not have access to the evaluation result of any intermediate layer, which might leak information about the model parameters \cite{FALCON_CVPR}. 

For the Conv layers, our proposed scheme also transmits the polynomials $p^{(n)}_1(X) = \rho^{(n)}(X) b(X)+ e^{(n)}_{1}(X)$ and $p_2^{(n)}(X) = \hat f^{(n)}(X)+\rho^{(n)}(X)+e_2^{(n)}(X)$ to the client. In our proposed packing scheme, shown in Fig. \ref{fig:batch_conv_example}, zeros are inserted into $\hat f^{(n)}(X)$ and their locations are known to the clients. 
\begin{thm}
    The nonzero coefficient of $\hat f^{(n)}(X)$ can not be recovered from either $p_1^{(n)}(X)$ or $p_2^{(n)}(X)$. 
\end{thm}
\begin{proof}
    Apparently, the coefficients of $\rho^{(n)}(X)$ can not be recovered from $p_2^{(n)}(X)$ when the corresponding coefficients in $\hat f^{(n)}(X)$ is nonzero.
    Next, it will be shown that $\rho^{(n)}(X)$ cannot be recovered from $p_1^{(n)}(X)$. Since the coefficients of $\rho^{(n)}(X)$ are sampled uniformly at random from $\mathcal{U}(0, \Delta)$ and the coefficients of $\hat{f}^{(n)}(X)$ also lie within the range $[0, \Delta]$, it follows that the nonzero coefficients of $\hat{f}^{(n)}(X)$, which are the kernel coefficients, cannot be recovered from $p_2^{(n)}(X)$ due to the secret sharing concept \cite{Gazelle}.

     Let us denote the $i$-th coefficient of a polynomial $a(X)$, by $a_i$. Suppose that $\hat f^{(n)}_i$ is zero. Then, the corresponding coefficient of $p^{(n)}_2(X)$ becomes
\begin{align*}
    p^{(n)}_{2,i} = \hat f^{(n)}_i + \rho^{(n)}_i + e^{(n)}_{2,i} = \rho^{(n)}_i + e^{(n)}_{2,i}.
\end{align*}
The coefficients of $e_2^{(n)}(X)$ follow a normal distribution with variance $\sigma^2$. If $e_{2,i}^{(n)}$ has large magnitude comparable to $\Delta$, then apparently $\rho^{(n)}_i$ can not be recovered from $p^{(n)}_{2,i}$. For the other case, assume that $e_{2,i}^{(n)}$ has $\delta$ bits. It only affects the $\delta$ least significant bits (LSBs) of $p^{(n)}_{2,i}$. Denote the other bits by most significant bits (MSBs). Then
\begin{align}
    \text{MSB}(p^{(n)}_{2,i}) = \text{MSB}(\rho^{(n)}_i + e^{(n)}_{2,i}) = \text{MSB}(\rho^{(n)}_i). \label{eq: MSB P_2}
\end{align}

 Break up $\rho^{(n)}(X)$ into the sum of $\rho^{(n)}_1(X)$ and $\rho^{(n)}_2(X)$ as follows
\begin{align*}
    \rho^{(n)}_{1,i} &= \begin{cases}
        \text{MSB}(\rho^{(n)}_i)\times 2^{\delta}  & \text{if } f^{(n)}_i = 0 \\
        0 & \text{otherwise}
        \end{cases}, \\
    \rho^{(n)}_{2,i} &= \begin{cases}
        \text{LSB}(\rho^{(n)}_i) & \text{if } f^{(n)}_i = 0 \\
        \rho^{(n)}_i & \text{otherwise}
        \end{cases}.
\end{align*}
For $f_i^{(n)}=0$, the MSB of $\rho^{(n)}_i$ are known from the MSB of $p^{(n)}_{2,i}$ using \eqref{eq: MSB P_2} and accordingly $\rho^{(n)}_{1,i}$ is known. However, the polynomial $\rho_2^{(n)}(X)$ is unknown. In particular, when $f_i^{(n)}\neq 0$, the coefficient $\rho_{2,i}^{(n)}$ is unknown. This is because that, substituting $\rho^{(n)}(x) = \rho^{(n)}_1(X) + \rho^{(n)}_2(X)$ into the formula of $p_1^{(n)}(X)$, it can be derived that 
\begin{align*}
    p^{(n)}_1(X) - \rho^{(n)}_1(X)b(X) = \rho^{(n)}_2(X)b(X) + e^{(n)}_1(X).
\end{align*}
Since the entropy of $\rho_2^{(n)}(X)$ is sufficiently high \parencite{LWE3}, it can not be recovered from $p_1^{(n)}(X)$ according to the hardness of the (R)LWE problem \parencite{LWE1, LWE2, RLWE1}. Moreover, since $\rho^{(n)}_i=\rho_{2,i}^{(n)}$ for the case of $f^{(n)}_i\neq 0$, $\rho^{(n)}_i$ can not be recovered from $p_1^{(n)}(X)$.
\end{proof}

Similarly, it can be proved that $p_1(X)$ and $p_2(X)$ used for evaluating the fully connected (FC) layers do not reveal any information about the model to the client.

From \eqref{eq:encryption}, noise is introduced to hide the message polynomial in the encryption process by adding the random polynomials $e_0(X)$ and $e_1(X)$ with variance $\sigma^2$. The proposed scheme does not change the noise that is added to $c_0(X)$ for encryption. According to Lines 6, 7, 16, 17, 27, 28 of Algorithm \ref{alg:proposed_ext} and \eqref{eq: c0+c1}, the noise of $r(X)$ computed in Line 32 of Algorithm \ref{alg:proposed_ext} equals
\begin{align}
   &\Delta^{-1}(\rho^{(n)}(\!X\!)e_1(\!X\!)\!+\!v(\!X\!)\rho^{(n)}(\!X\!)e(\!X\!)
    \!+\!v(\!X\!)e_{1}^{(n)}(\!X\!) \nonumber\\
    &\quad+i(X)e_2^{(n)}(X)+sh(X)e_2^{(n)}(X)). \nonumber
\end{align}
Following the heuristic analysis from \cite{CKKS} and based on the approach in \cite{AES, HE_noise}, the error in our protocol is bounded by
\begin{align}
     8N\sigma\sqrt{\frac{N}{6}}+ 8(\sqrt{3}+\sqrt{2})\sigma N. \label{eq: noise ours}
\end{align}
It is much lower than $\Delta$, and hence does not affect the correct evaluation result in $r(X)$. A similar analysis applies to the noise in our proposed FC layer evaluation scheme.

In the original scheme of ConvFHE \cite{ConvFHE}, the kernel plaintexts are first multiplied by the input ciphertexts. Following a similar analysis, its noise is bounded by
\begin{align}
     8N\sigma\sqrt{\frac{N}{3}} + 8 (\frac{1}{\sqrt{3}} + \sqrt{\frac{\eta}{3}})\sigma N, \label{eq: noise origin}
\end{align}
where $\eta$ is the Hamming weight of the secret key, which is typically 192 for various settings. Comparing \eqref{eq: noise ours} and \eqref{eq: noise origin}, our proposed approaches have lower noise. Additionally, ciphertext rotations are needed in ConvFHE to pack the valid slots into a single ciphertext. The multiplication with rotation keys further scale up the noise. 


\section{Experimental Results and Comparisons}
In this section, our proposed design is evaluated for one individual Conv layer as well as four variants of plain-20 classifier on CIFAR-10/100 datasets and ResNet-50 over ImageNet. Our design is compared with the classic Gazelle \parencite{Gazelle}, Cheetah \parencite{Cheetah}, and ConvFHE \parencite{ConvFHE}, which are among the most efficient existing designs. To facilitate a fair comparison with ConvFHE, which was originally implemented in Go, the Lattigo library  \parencite{lattigoLibrary}, which supports various HE schemes in Go, was utilized to integrate our proposed scheme into the existing framework. Gazelle and Cheetah were implemented in C++. All experiments were conducted by running a single thread on the Owens Cluster of Ohio Supercomputer Center \parencite{OSC} with 128GB of memory and an Intel Xeon E5-2680 V4 processor-based supercomputer. 


\subsection{Individual Layer Evaluation and ReLU Implementation}
To achieve 128-bit security and allow one level of multiplication between each bootstrapping, $Q$ is set to 104 bits for the CKKS scheme in our design following the HE parameters specified in \parencite{Bootstrapping4}. Besides, after each multiplication, the polynomial coefficients are scaled down by a factor to reduce noise, resulting in a new modulus $Q'$ with 55 bits. The same number of bits can be utilized for the polynomial coefficients in the Gazelle \parencite{Gazelle}, Cheetah \parencite{Cheetah}, and ConvFHE \parencite{ConvFHE} designs.


\subsubsection{Conv Layer Evaluation}
The complexities of carrying out one Conv layer with $c_i$ input channels, $c_o$ output channels, filter size $f_w=f_h=f$, input size $w_i=h_i=w$, and length-$N$ ciphertext polynomials on one input $i(X)$ using our proposed design are listed in Table \ref{tab: HE Conv Cost}. It is assumed that $N=c_iw^2$. In our design, the complexity is divided between the server and the client. All input data associated with one output channel is packed into one ciphertext and all associated filter coefficients are also packed into one plaintext polynomial. Besides, only the $c_0(X)$ element of the input data ciphertexts is transmitted to the server. 
\begin{table}[t]
\begin{center}\caption{Complexity of homomorphic evaluation of a Conv layer with $c_i$ input channels, $c_o$ output channels, input size $w_i=h_i=w$, filter size $f_w=f_h=f$, and length-$N=c_iw^2$ ciphertext polynomials.}
\label{tab: HE Conv Cost}
\begin{tabular}{@{}c@{}|@{}c@{}|@{}c@{}|@{}c@{}|@{}c@{}}
\hline
\diagbox[width=8em, height=3.5em]{Complexity}{Method} & Gazelle \parencite{Gazelle} & Cheetah \parencite{Cheetah} & ConvFHE \parencite{ConvFHE} & 
\begin{tabular}{@{}c@{}} 
Proposed
\\ 
(Server side,\\
Client side)
\end{tabular}\\ \hline
\# of NTTs & $4f^2$ & $2$ & $2$ & $(1, 3)$\\ \hline
\# of INTTs & $4c_o$ & $2c_o$ & $2c_o$ & $(c_o, c_o)$\\ \hline
\# of CWMs & $4f^2c_o$ & $2c_o$ & $2c_o$ & $(2c_o, 2c_o)$\\ \hline
\# of rotations & 
\begin{tabular}{@{}c@{}} 
$2(f^2\!-\!2\!+\!c_o)$  
\end{tabular}
& $0$ & $c_o-1$ & $(0,0)$ \\ \hline
\# of data in a poly. & $N/2$ & $N$ & $N$ & $N$ \\ \hline
\begin{tabular}{@{}c@{}} 
Memory (\# of
\\ 
length-$N$ poly.)
\end{tabular}
& 
\begin{tabular}{@{}c@{}} 
$f^2c_o+$\\
$2(f^2\!-\!2\!+\!c_o)$
\end{tabular}
& \begin{tabular}{@{}c@{}} $c_o$  \end{tabular} & \begin{tabular}{@{}c@{}} $2\log_2^{c_i}+c_o$ \end{tabular} & $(2c_o\!\!+\!\!1,2c_o\!\!+\!\!1)$ \\ \hline
\end{tabular}
\end{center}
\end{table}

    Using NTT, long polynomial multiplications can be simplified to coefficient-wise multiplications (CWMs). For the encryption, the NTTs of $a(X)$ and $b(X)$, which are the elements of the public key $pk$, are pre-computed. After NTT is applied to $v(X)$, CWMs are carried out between $NTT(b(X))$ and $NTT(v(X))$. The result is then added to $NTT(\Delta i(X))$ and $NTT(e_0(X))$ to derive $NTT(c_0(X))$. On the server side, the NTTs of $\rho^{(n)}(X)$ and $\hat{f}^{(n)}(X)$ are precomputed and shared among the evaluation of different inputs. The NTT of $sh(X)$ is also computed. Then, each of the $2c_o$ polynomial multiplications required in Line 16 of Algorithm \ref{alg:proposed_ext} is carried out as CWMs. On the other hand, INTT needs to be carried out to recover $c_0^{(n)}(X)$ before its coefficients are repacked into $c_0^r(X)$ in Line 20 of Algorithm \ref{alg:proposed_ext}. The server also computes the NTTs of $p^{(n)}_1(X)$ and $p^{(n)}_2(X)$ before they are sent to the client. Since $NTT(i(X))$ and $NTT(v(X))$ are already available from encryption on the client, the NTT of $c_1^{(n)}(X)$ specified in Line 27 of Algorithm \ref{alg:proposed_ext} can also be completed via CWMs. Similarly, INTT is applied to recover $c_1^{(n)}(X)$ before its coefficients are packed into $c_1^r(X)$. As discussed previously, only a subset of the $c_0^{(n)}(X)$ and $c_1^{(n)}(X)$ coefficients are needed, and their corresponding INTT can be simplified. For the polynomials whose values do not change with input or not re-generated with every input, including $a(X)$, $b(X)$, $p_1^{(n)}(X)$, $p_2^{(n)}(X)$, $\rho^{(n)}(X)$ and $\hat f^{(n)}(X)$, the NTTs only need to be computed once and stored. Hence, their complexities are excluded from Table \ref{tab: HE Conv Cost}.

On the server side, the NTTs of $\hat f^{(n)}(X)$ for filter coefficients, the random polynomials $\rho^{(n)}(X)$, and $sh(X)$ need to be stored. It is assumed that the NTTs of $p_1^{(n)}(X)$ and $p_2^{(n)}(X)$ are stored on the client. Besides, $NTT(v(X))$ must be retained from the time it is generated during the encryption process until it is used in the computations at Line 27 of Algorithm~\ref{alg:proposed_ext}. As a result, the memory requirements of our proposed design are listed as in Table~\ref{tab: HE Conv Cost}.  

The complexities of prior designs are also included in Table \ref{tab: HE Conv Cost}. In the proposed design, ciphertext rotations are entirely eliminated. Additionally, the complexity of polynomial multiplications is significantly reduced. The major reason is that only one instead of two polynomials are sent for each ciphertext for related computations and only one in every $c_i$ coefficient needs to be computed in the resulting product polynomials, as shown in Lines 20 and 31 of Algorithm \ref{alg:proposed_ext}. The polynomial multiplication complexity and memory requirement of Gazelle are in a higher order because they use a different packing scheme. Although the Cheetah design does not require rotation, it needs to compute every coefficient in the product polynomials and has a very high communication cost. The ConvFHE design packs the data for different output channels into the same ciphertext to address the communication cost issue of Cheetah, but it introduces ciphertext rotations and additional polynomial multiplications.

\begin{figure}[t]
		\centering
		\includegraphics[scale=0.7]{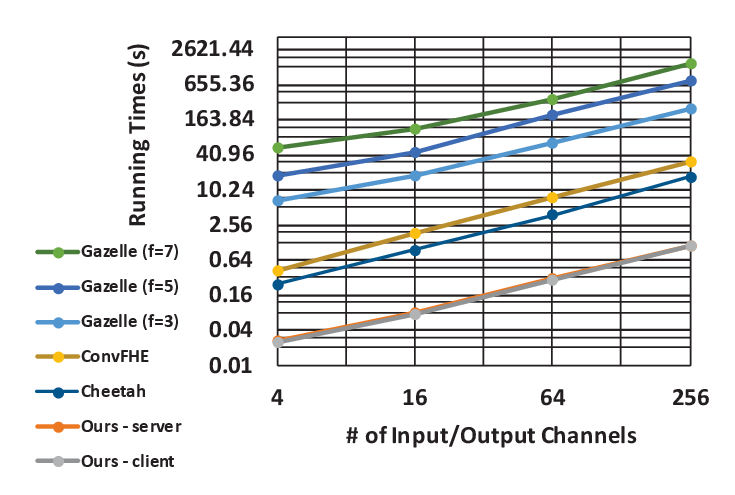}
		\caption{Latency of Conv layer evaluations.} \label{fig:batch_conv_compare}
\end{figure}

Our design for evaluating a single Conv layer is simulated and the latency is shown in Fig. \ref{fig:batch_conv_compare} for different number of input/output channels. In our simulations, $N=2^{16}$. It is assumed that $c_i=c_o$, $h_i=w_i=w$, and $c_iw^2=N$. Hence, $w$ is adjusted according to different $c_i$ in our simulations. The latencies of the computations on the server and client sides in our proposed design are shown separately in Fig. \ref{fig:batch_conv_compare} and they do not change with the filter size, $f$. For comparison, the Gazelle, Cheetah, and ConvFHE designs are also simulated over the same platform. Similar to that of our design, the latencies of Cheetah and ConvFHE schemes do not depend on $f$. However, the latency of Gazelle increases with larger $f$ because it packs different rows of the filters into separate polynomials. Besides, since Gazelle only packs $N/2$ data into one polynomial, it was iterated twice in the simulation to process the same amount of input data. Overall, the ConvFHE scheme exhibits a latency approximately $8\times$ to $13\times$ longer than that of our design on both the server and client side as $c_i/c_o$ increases from 4 to 256. Its longer latency is due to the expensive packLWE algorithm and many polynomial multiplications. Cheetah has around $5\times$ to $7.5\times$ longer latency than our design for $c_i/c_o$ in the same range. 
The Gazelle design for kernel sizes $f=3$, $f=5$, and $f=7$ have $134\times$, $358\times$, and $1057\times$, respectively, longer latency when $c_i=c_o=4$ compared to our proposed scheme due to the large number of ciphertext rotations. The relative latency is also becoming longer for larger $c_i/c_o$. Because of Gazelle’s high latency for linear layer evaluation, its performance is not considered in the evaluation of CNN classifiers in this paper.

\begin{table}[h]
\begin{center}\caption{Communication cost of evaluating ReLU function after Conv layer with $c_i$ input channels, $c_o$ output channels, data input size $w_i=h_i=w$, filter size $f_w=f_h=3$, $N=2^{13}$, $\lceil log_2Q\rceil=104$, and $\lceil log_2Q'\rceil=55$.}
\label{tab: HE ReLU comm cost}
\begin{tabular}{@{}c@{}|@{}c@{}|@{}c@{}|@{}c@{}|@{}c@{}}\hline
Parameters & Gazelle \parencite{Gazelle} &Cheetah \parencite{Cheetah} & ConvFHE \parencite{ConvFHE} & Proposed \\\hline \hline
$w,c_i,c_o$ & 
\begin{tabular}{@{}c@{}} 
$2N(\lceil\frac{2w^2c_i}{N}\rceil$
\\
$\lceil\log_2{Q}\rceil$ 
\\ 
$+\lceil\frac{2w^2c_o}{N}\rceil$ \\
$\lceil\log_2{Q'}\rceil)$
\end{tabular}
& \begin{tabular}{@{}c@{}} 
$2N(\lceil\frac{w^2c_i}{N}\rceil$ \\
$\lceil\log_2{Q}\rceil$\\ 
$+ c_o \lceil\log_2{Q'}\rceil)$ \end{tabular} 
& \begin{tabular}{@{}c@{}} 
$2N(\lceil\frac{w^2c_i}{N}\rceil$
\\
$\lceil\log_2 Q\rceil$
\\ 
$ +\lceil\frac{w^2c_o}{N} \rceil$
\\
$\lceil\log_2{Q'}\rceil)$ \end{tabular}
& 
\begin{tabular}{@{}c@{}} 
$N\lceil\frac{w^2c_i}{N}\rceil$ \\ $\lceil\log_2 Q\rceil$
\\
$+ \lceil w^2c_o \frac{c_o}{c_i}\rceil$ \\ $\lceil\log_2{Q'}\rceil$
\end{tabular} \\ \hline
$7,256,256$ & $1.3$MB & $29.26$MB & $0.65$MB & $0.3$MB\\ \hline 
$15,128,128$ & $2.6$MB & $15.27$MB & $1.3$MB & $0.62$MB\\ \hline 
$31,64,64$ & $5.21$MB & $8.91$MB & $2.61$MB & $1.27$MB \\ \hline 
$63,32,32$ & $10.42$MB & $7.01$MB & $5.21$MB & $2.58$MB \\\hline
\end{tabular}
\end{center}
\end{table} 

\subsubsection{ReLU Implementation}
In both our proposed design and Cheetah \parencite{Cheetah}, the ReLU function is evaluated using the millionaire protocol, which relies on VOLE-style OT for secure communication \cite{Cheetah, CrypTFlow2, VOLE-OT}. Gazelle \parencite{Gazelle} utilizes GCs to compute the non-linear layers. The ConvFHE design has two variations and can implement the ReLU using either high-order polynomial approximation or 2-PC protocols. Similar to Cheetah, ConvFHE evaluates ReLU using OT in the 2-PC setting. The 2-PC protocols take the decrypted result of the Conv layer as input, regardless of the evaluation scheme.  A significant amount of the communication between the client and server is associated with the ciphertexts transferred between the server and client. Table \ref{tab: HE ReLU comm cost} summarizes the communication costs under various parameter settings used in the Conv layers of the plain-20 and ResNet-50 classifiers.

The formulas in the second row of Table \ref{tab: HE ReLU comm cost} consist of two components each. The first represents the number of megabytes of data transmitted from the client to the server, and the second represents the data transmitted from the server to the client. To show the relative values, the communication costs for various $w$, $c_i$, and $c_o$ are listed in the other rows of Table \ref{tab: HE ReLU comm cost} for the case that the CKKS algorithm is adopted. As mentioned earlier, to achieve 128-bit security, the input polynomials have $\lceil log_2Q\rceil=104$-bit coefficients, and the output polynomials of the Conv layer have $\lceil log_2Q'\rceil=55$-bit coefficients. In our proposed design, the data from $\lfloor N/w^2\rfloor$ channels are packed into the same polynomial. The polynomials are long, and many slots remain unused when the channel number is small. When the number of channels is larger, the slots in the polynomials are more efficiently used, and the total number of polynomials involved is reduced. As a result, the communication cost of our design reduces as the number of channels increases.

\begin{table}[t]
\begin{center}\caption{Complexity of homomorphic evaluation of a FC layer with input size $n_i$, output size $n_o$, and length-$N$ ciphertext polynomial.}
\label{tab: HE FC Cost}
\begin{tabular}{c|@{}c@{}|@{}c@{}|@{}c@{}|@{}c@{}}
\hline
\diagbox[width=7.4em, height=3.5em]%
{\hspace*{-0.8em}\makebox[0pt][l]{Complexity}}%
{\makebox[0pt][r]{Method}\hspace*{-0.7em}} 
& Gazelle \parencite{Gazelle} & Cheetah \parencite{Cheetah} & ConvFHE \parencite{ConvFHE} &
\begin{tabular}{@{}c@{}} 
Proposed
\\ 
(Server side,\\
Client side)
\end{tabular}\\ \hline
\# of NTTs & 
\begin{tabular}{@{}c@{}}
     $2$
\end{tabular}
&
\begin{tabular}{@{}c@{}}
     $2$
\end{tabular}
&
\begin{tabular}{@{}c@{}}
     $2$
\end{tabular}
&
$(1,3)$
\\ \hline
\# of INTTs & 
\begin{tabular}{@{}c@{}}
     $2\lceil\frac{n_o}{\lfloor N/2n_i\rfloor}\rceil$
\end{tabular}
&
\begin{tabular}{@{}c@{}}
     $2\lceil\frac{n_o}{\lfloor N/n_i\rfloor}\rceil$
\end{tabular}
&
\begin{tabular}{@{}c@{}}
     $2\lceil\frac{n_o}{\lfloor N/n_i\rfloor}\rceil$
\end{tabular}
&
\begin{tabular}{@{}c@{}}
     $(\lceil\frac{n_o}{\lfloor N/n_i\rfloor}\rceil,$\\
     $\lceil\frac{n_o}{\lfloor N/n_i\rfloor}\rceil)$
\end{tabular}
\\ \hline
\# of CWMs & 
\begin{tabular}{@{}c@{}} 
$2N\times$
\\ 
 $\lceil\frac{n_o}{(\lfloor N/2n_i\rfloor)}\rceil$
\end{tabular} 
&
\begin{tabular}{@{}c@{}} 
$2N\times$
\\ 
 $\lceil \frac{n_o}{(\lfloor N/n_i\rfloor)}\rceil$
\end{tabular}
&
\begin{tabular}{@{}c@{}} 
$2N\times$\\
$\lceil \frac{n_o}{(\lfloor N/n_i\rfloor)}\rceil$
\end{tabular} & 
\begin{tabular}{@{}c@{}}
     $(2N\lceil\frac{n_o}{\lfloor N/n_i\rfloor}\rceil,$\\
     $2N\lceil\frac{n_o}{\lfloor N/n_i\rfloor}\rceil)$
\end{tabular}\\ \hline
\# of rotations & 
\begin{tabular}{@{}c@{}} 
$\lceil\frac{n_o}{(\lfloor N/n_i\rfloor)}$
\\
$-1+$
\\ 
$\log_2{\frac{N}{n_o}}\rceil$
\end{tabular}
& 
\begin{tabular}{@{}c@{}} 
$0$
\end{tabular}
& $2^{\lceil\!  \log_2{\!\frac{n_o}{(\!\lfloor\! N\!/\!n_i\!\rfloor\!)}\!}\!\rceil}$
& $(0,0)$ \\ \hline
\begin{tabular}{@{}c@{}} 
Memory \\
(\# of
\\ 
coeff.)
\end{tabular}
& 
\begin{tabular}{@{}c@{}} 
$N\times$
\\
$(\lceil\frac{3n_o}{(\lfloor N/n_i\rfloor)}$
\\ 
$-2+2\times$\\
$\log_2{\frac{N}{n_o}}\rceil)$
\end{tabular}
& 
\begin{tabular}{@{}c@{}} 
$N\times$
\\
$\lceil\frac{n_o}{(\lfloor N/n_i\rfloor)}\rceil$
\end{tabular}
& 
\begin{tabular}{@{}c@{}} 
$N\times$\\
$\lceil\frac{n_o}{(\lfloor N/n_i\rfloor)}\rceil$ \\
$+2N\times$\\
$\lceil\!  \log_2{\!\frac{n_o}{(\!\lfloor\! N\!/\!n_i\!\rfloor\!)}\!}\!\rceil$
\end{tabular}
&
\begin{tabular}{@{}c@{}} 
$(2N\lceil \frac{n_o}{(\lfloor N/n_i\rfloor)}\rceil,$
\\ 
$2N\lceil \frac{n_o}{(\lfloor N/n_i\rfloor)}\rceil$
\\
$+n_o)$
\end{tabular}\\ \hline
\end{tabular}
\end{center}
\end{table}

Compared to the ConvFHE scheme, our approach requires sending only one polynomial per ciphertext between the server and client, instead of two. Therefore, the communication cost is reduced to almost half. Gazelle packs half of the input channels into each polynomial compared to our design, and hence it requires around 4 times the communication cost. Cheetah sends $c_o$ instead of one ciphertext from the server to the client, leading to much higher communication overhead, which increases almost linearly with $c_o$ as $c_o$ grows.

\subsubsection{FC Layer Evaluation}
The complexity of evaluating one FC layer with an input size of $n_i$, output size of $n_o$, and a  length-$N$ ciphertext polynomial is presented in Table \ref{tab: HE FC Cost}. This complexity is derived using an analysis similar to that of the Conv layer. Notably, our design does not require ciphertext rotations. In our design, $\lfloor N/n_i\rfloor $ rows of the weight matrix are packed into a single polynomial $w(X)$. Therefore, the entire weight matrix can be packed into $\lceil n_o/(\lfloor N/n_i\rfloor)\rceil$ polynomials. 

The complexities of FC layer evaluation using the Gazelle \parencite{Gazelle}, Cheetah \parencite{Cheetah}, and ConvFHE \parencite{ConvFHE} designs are also listed in Table \ref{tab: HE FC Cost}. Overall, our design requires a smaller number of NTTs and INTTs compared to the prior works. Besides, Gazelle and ConvFHE require many expensive ciphertext rotations. Although our design demands approximately twice the storage compared to Cheetah and ConvFHE, this is not a major concern, as storage resources are relatively inexpensive.


\begin{table*}[t]
\begin{center}
\caption{Running time and communication cost comparisons of homomorphically encrypted variants of plain-20 and ResNet-50 classifiers over CIFAR-10/100 and ImageNet datasets, respectively, where $f$, $d$, and $w$ are the parameters of the classifiers indicating filter size, number of layers, and wideness factor, respectively.}\label{tab: CNN arcs result}
\begin{tabular}{c|@{}c@{}|c|c|c|c|c}
\hline
\multicolumn{2}{c|}{Classifier Architecture} & f3-d20-w1 & f5-d8-w3 & f3-d14-w3 & f3-d20-w3 & ResNet-50  \\
\hline\hline
\multirow{2}{*}{Accuracy (\%)} & 
\begin{tabular}{@{}c@{}} ConvFHE (approx. ReLU) \parencite{ConvFHE} \end{tabular} & 
\begin{tabular}{@{}c@{}} 90.32/64.08 \end{tabular} &
\begin{tabular}{@{}c@{}} 92.04/69.05 \end{tabular} &
\begin{tabular}{@{}c@{}} 93.99/73.47 \end{tabular} &
\begin{tabular}{@{}c@{}} 94.12/72.65 \end{tabular} & -
\\
\cline{2-7}
& \begin{tabular}{@{}c@{}} ConvFHE (OT ReLU) \& Proposed (OT ReLU) \end{tabular} &
90.39/64.13 &
92.37/69.57 &
94.04/73.65 &
94.30/72.95 & -
\\
\hline\hline
\multirow{8}{*}{Latency (s)} & 
Cheetah (Conv Layers without ReLU) \cite{Cheetah} & 38.2 & 68.1 & 150.1 & 232 & 1250
\\
& Cheetah (FC Layers without ReLU) \cite{Cheetah} & 1.8 & 1.9 & 1.9 & 2 & 8
\\
\cline{2-7}
& \begin{tabular}{@{}c@{}} 
ConvFHE (Conv Layers without ReLU)\parencite{ConvFHE}
\end{tabular} 
& 85.8 & 129.5 & 228.6 & 342.7 & -
\\
& \begin{tabular}{@{}c@{}} 
ConvFHE (FC Layers without ReLU)\parencite{ConvFHE} 
\end{tabular} 
& 4.2 & 4.5 & 4.4 & 4.3 & -
\\
\cline{2-7}
&\begin{tabular}{@{}c@{}} 
Proposed (Conv Layers without ReLU, Server)
\end{tabular} & 7 & 8.7 & 19.9 & 33.8 & 290 \\
&\begin{tabular}{@{}c@{}} 
Proposed (Conv Layers without ReLU, Client)
\end{tabular} & 7.3 & 8.5 & 20.2 & 34.5 & 292.7\\
&\begin{tabular}{@{}c@{}} 
Proposed (FC Layers without ReLU, Server)
\end{tabular} & 0.5 & 0.6 & 0.6 & 0.7 & 2 \\
&\begin{tabular}{@{}c@{}} 
Proposed (FC Layers without ReLU, Client)
\end{tabular} & 0.6 & 0.7 & 0.7 & 0.9 & 2.3 \\
\hline \hline
\multirow{6}{*}{Communication cost (MB)} & 
\begin{tabular}{@{}c@{}} Cheetah (OT ReLU, Conv Layers) \parencite{Cheetah} \end{tabular} & 301.9MB & 295.4MB & 570MB & 843.6MB & 10.97GB \\
 & 
\begin{tabular}{@{}c@{}} Cheetah (OT ReLU, FC Layers) \parencite{Cheetah} \end{tabular} & 4.3MB & 7.4MB & 8.9MB & 8.9MB & 3.2MB \\
\cline{2-7}
& 
\begin{tabular}{@{}c@{}} ConvFHE (OT ReLU, Conv Layers) \parencite{ConvFHE} \end{tabular} & 45.8MB & 22.5MB & 42.1MB & 60.8MB & - \\
& 
\begin{tabular}{@{}c@{}} ConvFHE (OT ReLU, FC Layers) \parencite{ConvFHE} \end{tabular} & 2MB & 2MB & 2MB & 2MB & - \\
\cline{2-7}
& \begin{tabular}{@{}c@{}} Proposed (Conv Layers) \end{tabular} 
& 23.1MB & 14.1MB & 26.6MB & 38.1MB & 640.2MB \\
& \begin{tabular}{@{}c@{}} Proposed (FC Layers) \end{tabular} 
& 1.6MB & 1.6MB & 1.6MB & 1.6MB & 0.8MB \\
\hline
\end{tabular}
\end{center}
\end{table*}

\subsection{Evaluation of CNN Classifiers}
To further evaluate our proposed design, it is applied to the plain-20 classifier over the CIFAR-10/100 datasets, and the ResNet-50 classifier over ImageNet \parencite{CNN_plain_20}. These classifiers are the deepest model tested in ConvFHE \parencite{ConvFHE} and Cheetah \cite{Cheetah}. The results are presented in Table \ref{tab: CNN arcs result}. Detailed information regarding the classifier parameters can be found in \parencite{CNNwide, CNN_plain_20}. As specified in the table, the classifiers were selected for their high accuracy.


Since our proposed design implements the ReLU function over decrypted data using OT, there is no accuracy loss compared to the original classifiers. The ConvFHE design approximates ReLU by performing computations on ciphertexts. Using higher-order polynomials for the approximation reduces accuracy loss, but also increases computational complexity. The latency of OT itself is negligible. However, the latency associated with transferring data between the client and server depends on various factors, such as network bandwidth and protocol. Hence, to evaluate the latency of the proposed scheme and those of previous designs implementing ReLU with OT, the running times of individual Conv and FC layers are separately added up and listed in Table \ref{tab: CNN arcs result}. It should be noted that the computations carried out on the client side overlap with those done on the server side. Hence, the running time of our proposed design may be shorter than the sum of the running time on the client and server sides, depending on the network latency and whether the server and client are available. Averaged over the four variants of plain-20 classifier, the running time of ConvFHE and Cheetah are $11.8\times$ and $6.7\times$, respectively, longer than that of our design. On the other hand, our proposed design also has a much lower communication bandwidth requirement. The data presented in Table \ref{tab: CNN arcs result} is the sum of the communication cost for individual Conv and FC layers, as well as the cost of evaluating ReLU through OT. It can be observed that the overall communication cost of our design is around 60\% and 6.7\%  compared to ConvFHE and Cheetah, respectively, using VOLE-style OT. Our design can also achieve a substantial reduction in the latency and communication cost of heavier networks. For instance, our design has $4.3\times$ shorter latency and $17.1\times$ lower communication cost compared to Cheetah for evaluating the ResNet-50 classifier. Since the ConvFHE did not report the latency of heavier networks, our design is not compared with ConvFHE for the ResNet-50 classifier. The communication latency depends on various factors, including bandwidth. As the primary contribution of this paper lies in optimizing the linear layers, our simulations specifically focus on these components. If end-to-end simulations, including communication overhead, were conducted across complete CNNs, the latency improvements achieved by our proposed design would likely be even more significant.

 
\section{Conclusions}
This paper proposed novel techniques to substantially reduce the complexity of homomorphically encrypted CNNs over the cloud. The expensive ciphertext rotations in the Conv and FC layers are completely eliminated by utilizing a reformulated server-client joint computation procedure and a new filter coefficients packing scheme without sacrificing the security of HE or increasing the noise level in the resulting ciphertexts. Besides, only selected coefficients from one polynomial are sent between the server and client for each ciphertext. This not only substantially reduces the complexity of the associated polynomial multiplications but also lowers the required communication bandwidth. Simulation results show that the running time of the Conv and FC layers in popular CNNs is reduced by at least a dozen times compared to the best prior design. Future research will focus on developing low-complexity methods for ReLU evaluation and further simplifying the bootstrapping process.

\printbibliography

\end{document}